\title{Average Stack Cost of {B\"uchi} Pushdown Automata\footnote{This work was supported 
by the National Science Centre (NCN), Poland under grant 2014/15/D/ST6/04543.}}

\documentclass[UKenglish]{article}
\usepackage{a4wide}
 
\usepackage{microtype}

\usepackage{authblk}
\author[1]{Jakub Michaliszyn}
\author[1]{Jan Otop}
\affil[1]{University of Wroc\l{}aw, Poland}

 \usepackage{amsthm,amsmath,amssymb,tikz,url}

\usepackage{thmtools,thm-restate}
\usetikzlibrary{patterns,shapes}

\newtheorem{theorem}{Theorem}
\newtheorem{lemma}[theorem]{Lemma}

\newenvironment{example}{\noindent\textbf{Example \refstepcounter{theorem}\arabic{theorem}.}}{}

\newcommand{\set}[1]{\{#1\}}

\newcommand{\aut}{\mathcal{A}}
\newcommand{\Q}{\mathbb{Q}}
\newcommand{\Z}{\mathbb{Z}}
\newcommand{\run}{\pi}
\newcommand{\cost}{\textbf{c}}

\newcommand{\pda}{(\Sigma, \Gamma, Q, Q_0, Q_F, \delta)}
\newcommand{\lang}{\mathcal{L}}

\newcommand{\WPS}{\mathcal{P}}

\newcommand{\letterCost}{\mathbf{lc}}
\newcommand{\tuple}[1]{\langle #1 \rangle}
\newcommand{\avgLC}{\mathsf{avg}\letterCost}
\newcommand{\avgInfLC}{\mathsf{avgInf}\letterCost}
\newcommand{\avgSupLC}{\mathsf{avgSup}\letterCost}

\newcommand{\ASC}{\mathsf{ASC}}
\newcommand{\IASC}{\mathsf{IASC}}
\newcommand{\SASC}{\mathsf{SASC}}

\newcommand{\IASCprob}[1]{\mathtt{IASC}}
\newcommand{\SASCprob}[1]{\mathtt{SASC}}

\newcommand{\avgInfLCprob}[1]{\mathtt{IALC}}
\newcommand{\avgSupLCprob}[1]{\mathtt{SALC}}

\newcommand{\WPScost}{\textrm{wt}}

\newcommand{\N}{\mathbb{N}}
\newcommand{\Paragraph}[1]{\noindent\textbf{#1.}}

\newcommand{\problem}[2]{\noindent$\blacktriangleright$ \textbf{#1}: #2}

\newcommand{\fst}{\mathit{first}_\run}
\newcommand{\lst}{\mathit{last}_\run}

\newcommand{\buchi}{B\"uchi{}}
\newcommand{\autMeta}{\aut^M}

\begin{document}

\maketitle

\begin{abstract}
We study the average stack cost of \buchi{} pushdown automata (\buchi{} PDA). We associate a non-negative price with each stack symbol and define the cost of a stack as the sum of costs of all its elements.
We introduce and study the average stack cost problem (ASC), which asks whether there exists an accepting run of a given \buchi{} PDA such that the long-run average of stack costs is below some given threshold.
The ASC problem generalizes mean-payoff objective and can be used to express quantitative properties of pushdown systems.
In particular, we can compute the average response time using the ASC problem.
We show that the ASC problem can be solved in polynomial time.
\end{abstract}

\section{Introduction}

Weighted pushdown systems (WPSs) combine finite-control,  unbounded stack and weights on transitions. Weights are aggregated using semiring operations~\cite{RepsSJM05} or the long-run average~\cite{ChV12}. 
These features make them a powerful formalism  capable of expressing interesting program properties~\cite{RepsSJM05,RepsFSTTCS}.  
Still, WPSs considered in the literature fall short of expressing the following basic quantitative specification.

Consider the following \emph{client-server scenario}, consisting of two agents, a server and a client. The client sends requests ($r$), which are granted ($g$) by the server. Each grant satisfies all pending requests.
All other events are abstracted to a null instruction ($\#$).
We are interested in checking properties of such systems over infinite runs. 
We are only interested in sequences with infinitely many requests and grants. 
The average workload property (AW) for client-server scenario, defined as the long-run average of the number of pending requests over all positions, was studied in~\cite{ChHO17Concur}.

WPSs can model the client-server scenario, but they cannot express AW for two reasons. 
First, WPSs considered in the literature~\cite{ChV12} have no \buchi{} acceptance condition, and hence we cannot specify traces with infinitely many requests and grants.  
Second, weights in WPSs are bounded, and hence the long-run average is bounded by the maximal weight, whereas AW is unbounded.

In this paper we study WPSs with \buchi{} acceptance conditions (known as \buchi{} pushdown automata) and unbounded weights depending on the stack content, called stack costs. 
More precisely, we define
the stack cost as a non-negative linear combination of the number of occurrences of every stack letter, i.e.,
given stack pricing that assigns a non-negative cost with stack symbols, the stack cost is the sum of prices of its elements.  
We investigate the \emph{average stack cost} (ASC) during infinite computations of an \buchi{} pushdown automaton. 
For a finite computation, the average stack cost is simply the sum of the stack costs in every position divided by the number of positions.
It is extended to infinite computations by taking the limit of the average stack costs of all the (finite) prefixes of this infinite computation.
As the limit may be undefined (when the sequence of prefixes diverge), we consider two values,  the limit inferior and limit superior over all prefixes.

We argue that with the ASC problem we can express interesting system properties. 
In particular, we can express  AW from the client-server scenario.
Moreover, we can express a variant of AW where each grant satisfies only one request.
This variant of AW cannot be specified with models from~\cite{ChHO17Concur}.
We can also use ASC to compute the average response time property~\cite{nested}, which asks for the average number of steps between a request and the corresponding grant. 
In this variant of the average response time, we can assume that each grant satisfies one request, which has not been possible in previous formalisms~\cite{nested}.

\Paragraph{Contributions} The main results presented in this paper are as follows.
\begin{itemize}
\item The average stack cost problem can be solved in polynomial time assuming unary encoding of stack pricing. 
\item One-player games on WPSs with the conjunctions of mean-payoff and \buchi{} objectives can be solved in polynomial time, even assuming binary encoding of weights.
\item The average response time property over WPSs in a variant of the client-server scenario where each grant satisfies only one request can be computed in polynomial time.
\end{itemize}

\Paragraph{Overview}
We start with basic definitions in Section~\ref{s:preliminaries}. 
Next, in Section~\ref{s:properties} we discuss convergence of the partial averages of the stack costs. 
In Section~\ref{s:reduction}, we show that to solve ASC we can bound the stack costs along the whole run.
This allows us to reduce ASC to the average letter cost problem, which is equivalent to 
one-player games on WPSs with the conjunction of mean-payoff and \buchi{} objectives.
We chose letter-based formalization rather than WPSs with weights on transitions, 
 as it allows us to use classical language-theoretic results on $\omega$-PDA.  
We apply these results in Section~\ref{s:lettercost} to show that the average letter cost problem can be solved in polynomial time.
Finally, we discuss the connection between ASC and the average response time property (Section~\ref{s:example}).

This is an extended version of the conference paper~\cite{MO17}.

\Paragraph{Related work}
WPSs with weights from a \emph{bounded idempotent semiring} and their applications have been studied in~\cite{RepsSJM05,RepsFSTTCS}.
In bounded idempotent semirings there are no infinite descending chains, e.g., the natural numbers, in contrast to the integers.
The results from~\cite{RepsFSTTCS} have been generalized to WPSs over indexed domains~\cite{minamide2013weighted}, which still do not capture the integers.
WPSs with integer weights aggregated with the long-run average operation (a.k.a. mean-payoff objective) have been studied in~\cite{ChV12}. 
It has been shown that one-player games on WPSs with mean-payoff objective can be solved in polynomial time.
 
The average stack cost is closely related to the average energy objective studied over finite graphs~\cite{bouyer2015average}.
In contrast to stack cost, energy levels are not observable, i.e., transitions do not depend on energy levels.
One player energy games are decidable in polynomial time. As we can express energy levels using stack costs, the results of this paper
can be considered as a generalization of the average-energy objective in the one-player case. 
However, two-player energy games are decidable in $\mathsf{NP} \cap \mathsf{coNP}$~\cite{bouyer2015average}, 
while even mean-payoff games on WPSs are undecidable~\cite{ChV12}. 
Since the average stack cost generalizes the mean-payoff objectives, two-player average-stack-cost games are undecidable.

\section{Preliminaries}
\label{s:preliminaries}

\Paragraph{Words and automata} 
 Given a finite alphabet $\Sigma$ of letters, a \emph{word} $w$ is a finite or infinite sequence 
of letters.
We denote the set of all finite words over $\Sigma$ by $\Sigma^*$, and the set of all infinite words over $\Sigma$ by $\Sigma^\omega$.
We use $\epsilon$ to denote the empty word.
 
For a word $w$, we define $w[i]$ as the $i$-th letter of $w$, and we define $w[i,j]$ as the subword $w[i] w[i+1] \ldots w[j]$ of $w$. 
We allow $j=\infty$ in $w[i,j]$. 
By $|w|$ we denote the length of $w$. 
We use the same notation for sequences that start from $0$.

A \emph{(non-deterministic) pushdown automaton} (PDA) is a tuple 
$\pda{}$, where 
	$\Sigma$ is the input alphabet, 
	$\Gamma$ is a finite stack alphabet, 
	$Q$ is a finite set of states, 
	$Q_0 \subseteq Q$ is a set of initial states,
	$Q_F \subseteq Q$ is a set of accepting states, and 
	$\delta \subseteq Q \times \Sigma \times (\Gamma \cup \{ \bot \}) \times Q \times \Gamma^*$
is a finite transition relation.
We define \buchi{}-PDA (called $\omega$-PDA for short) in the same way; these automata differ in semantics.
The size of an automaton $\aut = \pda$, denoted by $|\aut|$,
 is $|Q| + |\delta|$.

Assume a PDA (resp., $\omega$-PDA) $\aut=\pda$.
A \emph{configuration} of $\aut$ is a tuple $(q, a, u) \in Q\times (\Sigma \cup \set{\epsilon)} \times (\Gamma \cup \set{\bot})^*$, where $\bot$ occurs only once in $u$; it occurs as its first symbol.
A \emph{run} $\run$ of $\aut$ is a sequence of configurations such that
$\run[0] = (q_0, \epsilon, \bot)$ for some $q_0 \in Q_0$ and for every
$i < |\run|$, if $\run[i, i+1]=(q, a, u)(q', a', u')$, 
we have $\delta(q, a', x, q', y)$ for some $x, y$ such that either 
$x=u=\bot$ and $u'=y$ or $x \neq \bot$, $u=u_sx$ for some $u_s$ and $u'=u_sy$.
Runs of PDA are finite sequences, while runs of $\omega$-automata are infinite.

A run $\run=(q^0, a^0, u^0)(q^1, a^1, u^1) \dots $ \emph{gives} the word $a^0a^1\dots$.
A finite run $\run$ of a PDA is \emph{accepting} if the last state in $\run$ belongs to $Q_F$.
An infinite run $\run$ of an $\omega$-PDA is \emph{accepting} if it visits $Q_F$ infinitely often, i.e., satisfies the \buchi{} acceptance condition, and gives an infinite word.
The \emph{language recognized (or accepted) by the PDA $\aut$} (resp.,  $\omega$-PDA $\aut$), denoted $\lang(\aut)$, is the set of all words given by accepting runs of $\aut$.

\Paragraph{Weighted pushdown systems} A weighted pushdown system (WPS) $\WPS$ is pair $(\aut, \WPScost)$
such that 
(1)~$\aut$ is a PDA (resp.,  $\omega$-PDA) $\aut = \pda$,
(2)~the alphabet $\Sigma$ is a singleton,  
(3)~all states are accepting, i.e., $Q = Q_F$, and
(4)~$\WPScost$ is a cost function that maps transitions $\delta$ into a cost domain (which is $\Z$ in our case).
The alphabet $\Sigma$ and the set of accepting states are typically omitted.

\Paragraph{Context-free grammars (CFG) and their languages}  
A context-free grammar (CFG) is a tuple $G = (\Sigma,V,S,P)$, where $\Sigma$ is the alphabet, $V$ is a set of {\em non-terminals}, $S\in V$ is a {\em start symbol}, and $P$ is a set of {\em production rules}.
Each production rule $p$ has the following form $v\rightarrow u$, where $v\in V$ and $u \in (\Sigma \cup V)^*$.
We define \emph{derivation} $\rightarrow_G$ as a relation on $(\Sigma \cup V)^*\times (\Sigma \cup V)^*$ as follows:
$w \rightarrow_G w'$ iff $w = w_1 v w_2$, $w' = w_1 u w_2$, and  $v \rightarrow u$ is a production from $G$.
We define $\rightarrow_G^*$ as the transitive closure of $\rightarrow_G$. The \emph{language generated by $G$}, denoted by $\lang(G) = \{ w \in \Sigma^* \mid  S \rightarrow_G^* w \}$ is the set of words that can be derived from 
the start symbol $S$.  
CFGs and PDAs are language-wise polynomial equivalent (i.e., there is a polynomial time procedure that, given a PDA,  outputs a CFG of the same language and vice versa)~\cite{HU79}.

\subsection{Basic problems}
Let $\aut=\pda$ be an $\omega$-PDA. 
A \emph{stack pricing} is a function $\cost: \Gamma \to \mathbb{N}$ that assigns each stack symbol with a natural number (we assume $0$ is natural). 
We extend $\cost$ to configurations $(q, a, u)$ by setting $\cost((q,a,u)) = \sum_{i=1}^{|u|} \cost(u[i])$, where we assume that $\cost(\bot) = 0$.

Given a run $\run$ of $\aut$, a stack pricing $\cost$ and $k>0$, we define the \emph{average stack cost} of 
the prefix of $\run$ of length $k$, denoted by $\ASC(\run, \cost, k)$, as
$\frac{1}{k}\sum_{i=0}^{k-1} \cost(\run[i])$.

We are interested in establishing the average stack cost for the whole runs, which can be formalized in two ways. The \emph{infimum-average stack cost} of $\run$, denoted by $\IASC(\run, \cost)$, and the \emph{supremum-average stack cost} of $\run$, denoted by $\SASC(\run, \cost)$, are defined as
\[
\IASC(\run, \cost) = \liminf_{k\to\infty} \ASC(\run, \cost, k) \hspace{4em}
\SASC(\run, \cost) = \limsup_{k\to\infty} \ASC(\run, \cost, k)
\]

If $\cost$ is known from the contexts, we omit it and write $\IASC(\run)$ instead of $\IASC(\run, \cost)$ and similarly for $\SASC{}$.

We define two decision questions collectively called the average stack cost problem.

\problem{The $\IASCprob{\bowtie}$ problem}{given an $\omega$-PDA $\aut$, a stack pricing  $\cost$, $\bowtie \in \set{<, \leq}$ and a threshold $\lambda \in \Q$, decide whether
there exists an accepting run $\run$ of $\aut$ such that $\IASC(\run, \cost) \bowtie \lambda$.}

\problem{The $\SASCprob{\bowtie}$ problem}{given an $\omega$-PDA $\aut$, a stack pricing  $\cost$, $\bowtie \in \set{<, \leq}$ and a threshold $\lambda \in \Q$, decide whether
there exists an accepting run $\run$ of $\aut$ such that $\SASC(\run, \cost) \bowtie \lambda$.}

We assume that the numbers in the stack pricing and the threshold are given in unary, i.e., in an instance $I$ of the average stack cost problem, values of $\cost$ and $\lambda$ are polynomially bounded (in the size of the instance).

\noindent\emph{Remark}. 
Observe that the average stack cost problem generalizes WPSs with mean-payoff objectives. 
First, WPSs consists of a PDA (resp., $\omega$-PDA) and a cost function from transitions into integers.
We can however add a constant $C$ to all weights, which change all mean-payoff values by~$C$. Thus, we can assume that costs are non-negative.   
Second, we can emulate costs on transitions by extending the stack alphabet with letters corresponding to transitions, and storing the last taken transition at the top of the stack. 
Hence, allowing, in addition to stack costs, costs on transitions does not change the expressive power or the complexity. 
For simplicity, we do not consider costs on transitions.
Finally, the average stack cost strictly generalizes WPSs with mean-payoff objectives as it can be unbounded  whereas the mean-payoff is bounded by the maximal weight of the transition.

\begin{example}
Recall the client-server scenario from the introduction. Assume that the stack alphabet is $\Gamma = \set{r}$ and all requests are pushed on the stack by the client.
Then, upon a grant the server empties the stack. 
Observe that if the cost of a request on the stack is $1$, i.e., $\cost(r) = 1$, then the average stack cost equals AW.

We can modify this example to model that each grant satisfies a single request. Simply, we require the server to pop only a single request upon a grant.
Again, the average stack cost equals AW. 
\end{example}
 \section{Properties of Average Stack Cost}
\label{s:properties}
We now extend the notion of the average stack cost to automata, defining: \\
{\color{white}.}$ \IASC(\aut, \cost) = \inf \set{\IASC(\run, \cost) \mid \run \text{ is an accepting run of } \aut}$\\
$\SASC(\aut, \cost) = \inf \set{\SASC(\run, \cost) \mid \run \text{ is an accepting run of } \aut}$.

We can easily construct a run $\run$ and stack pricing $\cost$, such that $\IASC(\run, \cost) < \SASC(\run, \cost)$. 
We now show an example proving a stronger claim, stating that even $\IASC(\aut, \cost)$  and $\SASC(\aut, \cost)$  can have different values.

\begin{example}
Consider an automaton $\aut$ with three states $U, B, A$, one alphabet symbol $a$ and two stack symbols $\alpha, \beta$, and the stack pricing such that $\cost(\alpha)=0$ and $\cost(\beta)=3$. 
State $A$ is the only accepting and the only starting state. The transition function is as follows.
\begin{align*}
&\delta(A, a, \bot, U, \alpha)&
&\delta(U, a, \alpha, U, \alpha\alpha)&
&\delta(U, a, \alpha, B, \beta) \\
&\delta(B, a, \beta, B, \epsilon)&
&\delta(B, a, \beta, A, \epsilon)&
&\delta(B, a, \alpha, B, \beta)
\end{align*}

Every accepting run of $\aut$ starts in the state $A$, adds some number of symbols $\alpha$ to the stack in state $U$, and then goes to the state $B$, where it clears the stack, but to remove a symbol $\alpha$, it first needs to convert it to (costly) $\beta$. Then it reaches $A$ with empty stack and repeats.

Observe that $\SASC(\aut, \cost)=1$. 
To see this, consider an accepting run $\run$. For any position $p>0$ with an accepting state  we have that $\ASC(\run, \cost, p-1)=1$. To show this, we assign to every $\beta$ symbol that occur in $\run[0, p-1]$ three positions: right before it was removed, right before it replaced some $\alpha$, and right before this $\alpha$ was added. 
In this way we cover all the positions in $\run[0, p-1]$, which means that the number of $\beta$ symbols is three times the number of positions, so $\ASC(\run, \cost, p-1)=1$.

In contrast, we show that $\IASC(\aut, \cost)=0$. 
Let $\run_i$ be the sequence of configurations 
\[\hspace{-7pt} (A, a, \bot), (U, a, \bot\alpha),  \dots, (U, a, \bot\alpha^i), (B, a, \bot\alpha^{i-1}\beta), (B, a, \bot\alpha^{i-1}), (B, a, \bot\alpha^{i-2}\beta) \dots (B, a, \bot\beta)\]

For $\IASC$, consider 
a sequence $a_i$ defined recursively
 as $a_1=1$, $a_{i+1}=i \cdot \sum_{
 j=1}^i a_j$
 and a run $\run=\run_{a_1}\run_{a_2}\run_{a_3}\dots$.
For each $i$, we have $\sum_{j=0}^{a_1+\dots+a_i} \cost(\run[j]) = 3(a_1+\dots+a_i)$
(as in the $\SASC$ case, one can assign exactly three positions to each $\beta$). 
Therefore, at the position $3a_i+a_{i+1}$ in $\run$, which is in $\run_{a_{i+1}}$ and it is the first position there with $B$,
  the value $\ASC(\run, \cost, 3a_i+a_{i+1})$ ) can be bounded by $\frac{3(a_1+\dots+a_i)}{3(a_1+\dots+a_i)+a_{i+1}} = \frac{3(a_1+\dots+a_i)}{(3+i)(a_1+\dots+a_i)} = \frac{3}{3+i}$.
The sequence $\frac{3}{3+i}$ converges to $0$, and since we only have non-negative costs,  
$\IASC(\aut, \cost)=0$.
\end{example}

The above example uses both non-accepting states (to ensure that the stack is emptied infinitely often) and zero costs. 
Both are needed; we show a no-free-lunch theorem, saying that we can only have two out of three things: (1)~$\omega$-PDA with non-accepting states, (2)~stack symbols with cost $0$, 
or (3)~a guarantee that $\IASC$ and $\SASC$ coincide.

\begin{restatable}{theorem}{Convergance}
Let $\aut$ be an $\omega$-PDA and $\cost$ be a stack pricing $\cost$.
If $\aut$ has only accepting states or $\cost$ returns only positive values, then
$\IASC(\aut, \cost) = \SASC(\aut, \cost)$.
\end{restatable}

\begin{proof}[Proof sketch] 
In the only-accepting-runs case, the theorem follows from a reduction to one-player games on WPSs with 
mean-payoff objectives~\cite{ChV12}. Mean-payoff objectives are considered in two variants: mean-payoff infimum corresponding to limit infimum of partial averages and
mean-payoff supremum corresponding to  limit supremum of parial averages. However, it is shown in~\cite{ChV12} that winning against one objective (say, the mean-payoff infimum objective) is equivalent to 
winning against the other (the mean-payoff supremum objective). From that and our reduction, we conclude that $\IASC(\aut, \cost) = \SASC(\aut, \cost)$.
 In the only-positive-values case, we prove that it is enough to consider runs with bounded size of the stack; then, 
we reduce the average stack cost to the regular language case, in which the results on weighted automata~\cite{quantitativelanguages} and simple arguments show that 
the infimum over all runs of a weighted automaton with accepting states is realized by a run, in which partial averages converge. We conclude that
 $\IASC(\aut, \cost) = \SASC(\aut, \cost)$.
\end{proof}

We conclude with a realisability theorem, stating that if $\IASC(\aut, \cost)$ and $\SASC(\aut, \cost)$ coincide, then there is a single run that witnesses both.

\begin{restatable}{theorem}{Realisability}
For an $\omega$-PDA $\aut$ and a stack pricing $\cost$ we have $\IASC(\aut, \cost) = \SASC(\aut, \cost)$ iff there is a run $\run$ such that $\IASC(\run, \cost) = \SASC(\run, \cost)= \IASC(\aut, \cost)$.
\end{restatable}

\begin{proof}[Proof sketch] 
We prove that $\SASC(\aut, \cost)$ is always realized, i.e., for every PDA $\aut$ there exists $\pi$ such that $\SASC(\run, \cost)= \SASC(\aut, \cost)$. This immediately implies the theorem.
\end{proof}
 \section{From Average Stack Cost to Average Letter Cost }
\label{s:reduction}
The \emph{average letter cost problem} takes an $\omega$-PDA $\aut$ and a cost function defined on letters, and 
asks whether there is a word in the language of $\aut$ whose long-run average of costs of letters is below a given threshold.
This section is devoted to a polynomial time reduction from the average stack cost problem to 
the average letter cost problem. 

The reduction consists of two steps. First, we show that in the average stack cost problem, we can impose a bound $B$ on the stack cost (which depends on the $\omega$-PDA and the threshold). 
Next, we take the $\omega$-PDA $\aut$ from the average stack cost problem and define an $\omega$-PDA $\autMeta$, which 
recognizes words encoding the runs of $\aut$. The words accepted by $\autMeta$ correspond precisely to runs with stack height bounded by $B$ and 
are annotated with the current stack cost along the run. These annotated costs are treated as costs of the letters, which completes the reduction.

Formally, a letter-cost function $\letterCost$ is a function from a finite alphabet of letters $\Sigma$ into rationals. We assume the binary encoding of numbers. 
The letter-cost function extends naturally to words by $\letterCost(a_1 \ldots a_n) = \letterCost(a_1) + \ldots + \letterCost(a_n)$.
For a finite word $w$, we define the average letter cost $\avgLC(w)$ as $\frac{\letterCost(w)}{|w|}$. 
The average letter cost extends to infinite words as the low and the high average letter cost.
For an infinite word $w$, we define 
the average low letter cost as $\avgInfLC(w) = \liminf_{k\to\infty}  \avgLC(w[1,k])$ and the average high letter cost as $\avgSupLC(w) = \limsup_{k\to\infty} \avgLC(w[1,k])$).

\problem{The $\avgInfLCprob{\bowtie}$ problem}{given an $\omega$-PDA $\aut$, 
a letter-cost function $\letterCost$, $\bowtie \in \set{<, \leq}$ and a threshold $\lambda \in \Q$, decide whether
there exists a word $w \in \lang(\aut)$ such that $\avgInfLC(w) \bowtie \lambda$.}
\problem{The $\avgSupLCprob{\bowtie}$ problem} {given an $\omega$-PDA $\aut$, 
a letter-cost function $\letterCost$, $\bowtie \in \set{<, \leq}$ and a threshold $\lambda \in \Q$, decide whether
there exists a word $w \in \lang(\aut)$ such that $\avgSupLC(w) \bowtie \lambda$).}

In contrast to the average stack cost problem, we allow the binary encoding of numbers for $\letterCost$ and $\lambda$ (a rational is encoded as a pair of integers, which are encoded in binary).
The main result of this section is the following theorem:

\begin{theorem}\label{t:reduction}	
There are polynomial-time reductions from the $\IASCprob{}$ problem to the $\avgInfLCprob{\bowtie}$ problem  and from the $\SASCprob{}$ problem to the $\avgSupLCprob{\bowtie}$ problem.
\end{theorem}

We start the proof with auxiliary tools and lemmas.

\subsection{Pumping lemma}
For a run $\run$ and a position $i$, by $q_\run[i]$, $a_\run[i]$ and $u_\run[i]$ we denote the state, letter and stack at position $i$ of $\run$, i.e., $(q_\run[i], a_\run[i], u_\run[i]) = \run[i]$.

Consider a run $\run$. We define two useful  functions, $\fst(i,j)$ and $\lst(i,j)$, that take a 
position $i$ in $\run$ and a stack position $j$ of the configuration $\run[i]$ and return a position from $\run$. 
Intuitively, $\run[\fst(i,j)]$ is the configuration where the $j$th stack symbol in the $i$th configuration of $\run$ was added to the stack and $\run[\lst(i, j)]$ is the configuration right before this stack symbol was removed from the stack. More formally, the functions  $\fst$ and $\lst$ are such that for each $i\in \N$ and each $j\in \set{1, \dots, |u_\run[i]|}$, $\fst(i, j)$ is a minimal number and $\lst(i, j)$ is a maximal number such that $\fst(i,j)\leq i \leq \lst(i,j)$ and all the stacks among $u_\run[\fst(i,j)], \dots, u_\run[\lst](i,j)$ start with the same $j$ stack symbols $u_1, \dots, u_j$. 

A stack position $j$ in a configuration $i$ is \emph{persistent} if $\lst(i, j)=\infty$ (i.e., this symbol is never removed) and \emph{ceasing} otherwise.
We define a function \emph{lifespan} $ls_\run(i, j)=(\fst(i,j), \lst(i,j))$. 
For a finite word $w=w_1w_2, \dots, w_s$ let $w[l, \infty]$ denote the suffix $w_l w_{l+1} \dots w_s$.

Assume a run $\run$ and a position $i\in \N$ such that $u_\run[i]=u_1 \dots u_n$. Two stack positions $j, k \in \set{1, \dots, n}$ are \emph{equivalent in $\run[i]$} if 
\begin{itemize}
\item $u_j=u_k$ and they first appeared with the same symbols above, i.e.,
 $u_\run[\fst(i, j)][j, \infty] = u_\run[\fst(i, k)][k, \infty]$, and
\item $q_\run[\fst(i, j)]=q_\run[\fst(i, k)]$, and
\item Either both $j$ and $k$ are persistent in $i$, or both are ceasing and then $q_\run[\lst(i,j)]=q_\run[\lst(i,k)]$.
\end{itemize}

\begin{figure}
\centering

\begin{tikzpicture}[scale=1.4,permBox/.style={rectangle, draw, minimum  size=0.3cm, fill=green},
moveBox/.style={rectangle, draw, minimum  size=0.3cm,  fill=blue},
vanPosBox/.style={rectangle, draw, minimum  size=0.3cm, thin,gray},
vanStackBox/.style={rectangle, draw, minimum  size=0.3cm, thin,gray}]
\definecolor{green}{rgb}{0.0,0.6,0.0}
\definecolor{blue}{rgb}{0.4,0.8,1}

\foreach \x/\ys/\yf/\c in {1/1/1/black,2/1/2/black,11/1/2/black,12/1/1/black}
{
   \foreach \y in {\ys,...,\yf}
   {
      \node[permBox] at (\x*0.3,\y*0.3) {};   
   } 
}

\foreach \x/\ys/\yf/\c in {3/1/3/blue,4/1/4/blue,9/1/4/blue,10/1/3/blue}
{
   \foreach \y in {\ys,...,\yf}
   {
      \node[vanPosBox] at (\x*0.3,\y*0.3) {};   
   } 
}

\foreach \x/\ys/\yf/\c in {5/5/5/green,6/5/6/green,7/5/6/green,8/5/5/green}
{
   \foreach \y in {\ys,...,\yf}
   {
      \node[moveBox] at (\x*0.3,\y*0.3) {};   
   } 
}

\foreach \x in {5,...,8}
\foreach \y in {1,2}
{
  \node[permBox] at (\x*0.3,\y*0.3) {}; 
   \node[vanStackBox] at (\x*0.3,0.6+\y*0.3) {};   
}

\node (I) at (6*0.3,-0.2) {$i$};
\node (J) at (4.3,0.6) {$j$};
\node (K) at (4.3,1.2){$k$};

\draw[densely dashed] (3.5,0.6) to (J);
\draw[densely dashed] (3.0,1.2) to (K);

\node[rectangle, rounded corners=3, minimum size=0.82cm,draw, densely dashed, thick] (C1) at (0.45,0.45) {};
\node[rectangle, rounded corners=3, minimum size=0.82cm,draw, densely dashed, thick] (C2) at (3.45,0.45) {};

\node[rectangle, rounded corners=3, draw, minimum height=0.82cm, minimum width=1.65cm, densely dashed, thick] (E1) at (1.95,0.45) {};
\node[rectangle, rounded corners=3, draw, minimum height=0.82cm, minimum width=1.68cm, densely dashed, thick, red] (E2) at (1.95,1.65) {};

\begin{scope}[yshift=-2cm, xshift=0.6cm]

\foreach \x/\ys/\yf/\c in {1/1/1/black,2/1/2/black,7/1/2/black,8/1/1/black}
{
   \foreach \y in {\ys,...,\yf}
   {
      \node[permBox] at (\x*0.3,\y*0.3) {};   
   } 
}

\foreach \x/\ys/\yf/\c in {3/3/3/green,4/3/4/green,5/3/4/green,6/3/3/green}
{
   \foreach \y in {\ys,...,\yf}
   {
      \node[moveBox] at (\x*0.3,\y*0.3) {};   
   } 
}

\foreach \x in {3,...,6}
\foreach \y in {1,2}
{
  \node[permBox] at (\x*0.3,\y*0.3) {}; 

}

\node[rectangle, rounded corners=3, minimum size=0.82cm,draw, densely dashed, thick] (C11) at (0.45,0.45) {};
\node[rectangle, rounded corners=3, minimum size=0.82cm,draw, densely dashed, thick] (C21) at (2.25,0.45) {};

\node[rectangle, rounded corners=3, draw, minimum height=0.82cm, minimum width=1.62cm, densely dashed, thick] (E11) at (1.35,0.45) {};
\node[rectangle, rounded corners=3, draw, minimum height=0.82cm, minimum width=1.68cm, densely dashed, thick, red] (E21) at (1.35,1.05) {};

\end{scope}

\draw[->,thick] (C1) to[bend right] (C11);
\draw[->,thick] (C2) to[bend left] (C21);

\draw[->,thick] (E1.south west) to[bend right=60] (E11.north west);
\draw[->,thick] (E2.south east) to[bend left] (E21.north east);

\end{tikzpicture}

 \caption{An example of $i, j, k$-contraction. The top part of the picture illustrates twelve consecutive stack contents; 
the sixth one is marked as $i$ and contains two distinguished equivalent stack positions $j$ and $k$. 
The bottom part of the picture is obtained by removing configurations 3, 4, 9 and 10 and removing stack positions 3 and 4 at positions 5-8 in the run. }
\end{figure}
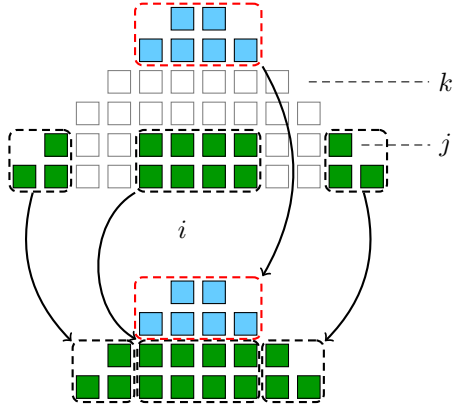

Assume a run $\run$, $i\in \N$ and two stack positions $j < k$ equivalent in $\run[i]$.
We define a $i, j, k$-contraction of $\run$ as a sequence $\run^C$ defined as follows:
\begin{itemize}
\item If $j$ and $k$ are ceasing, then $\run^C=\run[0, \fst(i,j)]\run'\run[\lst(i, j), \infty]$, where $\run'$ is the result of removing in $\run[\fst(i,k)+1, \lst(i,k)-1]$ in each stack symbols at positions $j+1, j+2, \dots, k$.
\item If $j$ and $k$ are persistent, then $\run^C=\run[0, \fst(i,j)]\run'$, where $\run'$ is the result of removing in $\run[\fst(i, k)+1, \infty]$ in each stack symbols at positions $j+1, j+2, \dots, k$.
\end{itemize}

The proof of the following lemma is now straightforward.

\begin{lemma} 
A contraction of an accepting run is an accepting run.
\end{lemma}

\newcommand{\stackmin}{3|Q| \cdot |\Gamma| \cdot |\delta|}
If the stack size is at least $\stackmin$, then either there are more than $2|\delta|$ persistent positions on the stack with the same stack symbol or more than $2|Q||\delta|$ ceasing positions with the same stack symbols; in both cases, one can always pick three equivalent positions among them. We state this observation as a lemma.

\begin{lemma}\label{l:bigstack}
Among any $\stackmin$ stack positions in any configuration $\run[i]$ there are three stack positions pairwise equivalent in $\run[i]$.
\end{lemma}

\subsection{The bounded stack cost property}

We show the bounded stack cost property for the $\IASCprob{}$ and $\SASCprob{}$ problem.

\newcommand{\maxdep}{\max_{s \in \Gamma}{\cost(s)} \cdot \stackmin + \lambda}
\begin{restatable}{lemmaStatement}{Bounded}
\label{l:bounded-cost}
Assume an $\omega$-PDA $\aut$, stack pricing $\cost$, $\bowtie \in \set{<,\leq}$, $\lambda \in \Q$ and an accepting run $\run$ such that $\IASC(\run) \bowtie \lambda$.  
There is an accepting run $\run'$ such that $\IASC(\run') \bowtie \lambda$ and for each $i$, $\cost(\run[i]) \leq \maxdep$. The same holds for $\SASC$.
\end{restatable}

\begin{proof}[Proof sketch]
We first show the proof for $\IASC$ and $\bowtie=\leq$.  Assume an $\omega$-PDA $\aut$, stack pricing $\cost$, $\lambda \in \Q$ and an accepting run $\run$ such that $\IASC(\run) \leq \lambda$. 

Let $i$ be the smallest number such that $\cost(\run[j]) \leq i$ for infinitely many $j$. 
Clearly $i \leq \lambda$ since $\IASC(\run) \leq \lambda$.
If there are only finitely many positions where the stack cost exceeded $\maxdep$, then for every such a position $i$ we can find, by Lemma~\ref{l:bigstack}, two equivalent stack positions $j<k$ and obtain the $i, j, k$-contraction of the run. We repeat it until the cost reaches the desired bound. Since we repeat this only finitely many times for the whole run, the obtained run $\run'$ is accepting
and $\IASC(\run')\leq \lambda$.

For the rest of this proof, we focus on the case where there are infinitely many positions with costly stack. There are two new challenges now: we need to make sure to preserve infinitely many accepting states, and guarantee that $\IASC$ stays within desired bound. 

To preserve infinitely many accepting states, we decompose $\run$ as $\run^{ok}_1 \run_1 \run^{ok}_2 \run_2 \run^{ok}_3 \dots$ such that 
every positions in $\run_i$ has stack cost exceeding $\lambda$ and all the positions of $\run^{ok}_i$ cost at most $\lambda$. 
Our goal is to define a new run $\run'$ that is obtained from $\run$ by contracting some of the runs among $\run_1, \run_2, \dots$.  To preserve the acceptance condition, we mark one configuration with an accepting state in each $\run_i$ that has such a configuration and guarantee that this state will be retained.

Let $\run'=\run^{ok}_1 \run'_1 \run^{ok}_2 \run'_2 \run^{ok}_3 \dots$, where for each $i$, we define $\run'_i$ starting from $\run_i$, and then by repeating the following procedure as long as needed. 
For any $j$ such that $\cost(\run_i'[j])>\maxdep$, there are at least $\stackmin$ stack positions $k$ whose symbols have positive cost and lifespan is within $\run'$, as the total cost of the remaining stack position is bounded by $\lambda$. 
Among them, we can find three pairwise equivalent stack positions, from which one can choose two positions $k, l$ such that the $j, k, l$-contraction of $\run'_i$ retains the marked accepting state. We set $\run'_i$ to be the contraction.

For each position $i$ in $\run'$, we define its origin $o(i)$ as the position in $\run$ from which $i$ originates ($o$ is a monotonic function). 
We argue that if $\ASC(\run, \cost, i) \leq \lambda$, then $\ASC(\run', \cost, o(i)) \leq \lambda$; the proof is based on the fact that we only remove or alter positions where the cost is greater than $\lambda$.
It follows that $\IASC(\run') \leq \IASC (\run)$, as required.
 
For the strict inequality assume that $\IASC(\run) < \lambda$. Then for some $\epsilon >0$, we have $\IASC(\run) \leq \lambda - \epsilon$. By the above reasoning, there exists $\run'$ 
such that $\IASC(\run') \leq \lambda - \epsilon$ and for each $i$, $\cost(\run[i]) \leq \maxdep$. Then, $\IASC(\run') < \lambda$.

The case of $\SASC$ uses the same construction, but the reasoning now is slightly more technical as we have to argue that all the subsequences of $\run'$ have the cost less than $\lambda$, but the idea is exactly the same, so we skip it here.
\end{proof}

We now prove Theorem \ref{t:reduction}.

\begin{proof}
We show the reduction of $\IASCprob{<}$ to $\avgInfLCprob{<}$. 
The reduction is through a construction of a \emph{meta-automaton} defined below. 
Given an instance  $I = \tuple{\aut, \cost, \bowtie, \lambda}$ of the $\IASCprob{<}$ problem, we define 
a \emph{meta-automaton} $\autMeta$ for $I$ as an  $\omega$-PDA that recognizes the language of infinite words corresponding to accepting runs of $\aut$.
Formally, let $N = \maxdep$. The $\omega$-PDA $\autMeta$ works over the alphabet $\delta \times \set{0,\ldots,  N}$, where $\delta$ is the transition relation of $\aut$, and
$\autMeta$ accepts all words $w$ such that 
(1)~$w$ encodes an accepting run $\pi^w$ of $\aut$, and
(2)~the stack cost of $\pi$ at position $i$ is encoded as the second component of $w[i]$.
In particular, (2) implies that the meta-automaton accepts words that correspond to runs of $\aut$ in which the stack cost is bounded by $N$ at every position.

To build a meta-automaton we need to track the current stack cost. However, even a finite-state automaton can store in its states the current stack cost, which belongs to
$\set{0,\ldots,  N}$ and update it based in the current symbol being pushed to or popped from the stack. Therefore, a meta-automaton can be constructed in polynomial time.

Consider a letter-cost function defined over $\delta \times \set{0,\ldots, N}$ such that 
$\letterCost(t,x) = x$, i.e., the letter cost of a transition is the current stack cost. 
Observe that each partial averages of stack cost in a run $\pi$ coincides with the corresponding partial average of letter costs in the corresponding word. 
Therefore, if the instance $(\autMeta,\letterCost,\bowtie, \lambda)$ of the $\avgInfLCprob{<}$ problem is solved by a word $w$, then the corresponding run $\run^w$ of $\aut$ is a solution of $I$.
Conversely, if $I$ has a solution $\run$, then by  Lemma~\ref{l:bounded-cost} it has a solution $\run'$, in which
all stack costs are bounded by $N$. Then, there is a word $w'$ accepted by $\autMeta$ corresponding to the run $\pi'$.
Observe that $w'$ is a solution of the instance $(\autMeta,\letterCost,\bowtie, \lambda)$ of the $\avgInfLCprob{<}$ problem.

The same construction gives us the reduction from $\SASCprob{\leq}$ to $\avgSupLCprob{\leq}$.
The prove of correctness is the same as we only need to use different variant of 
Lemma~\ref{l:bounded-cost}. 
\end{proof}
 
\section{The average letter cost problem}
\label{s:lettercost}
We prove that the average letter cost problem can be solved in polynomial time. 
In Section~\ref{s:finiteAWC} we study the finite-word variants of the average letter cost problem.
Next, we use finite-word results to solve the average letter cost problem over infinite word (Section~\ref{s:infiniteAWC}).
We supplement this section with the comparison of the average letter cost problem and one-player games on WPSs with 
conjunctions of mean-payoff and \buchi{} objectives.

\subsection{Average letter cost over finite words}
\label{s:finiteAWC}
\newcommand{\letterCostLambda}{\letterCost^{\lambda}}

We are interested in the average letter cost over all (finite) words accepted by a given PDA. 

\problem{The $\avgLC{}$ problem}{given a letter-cost function $\letterCost$, a PDA $\aut$, $\bowtie \in  \set{<, \leq}$ and a threshold $\lambda$, 
decide whether $\inf_{w \in \lang(\aut)} \avgLC(w) \bowtie \lambda$.}

To solve this problem, we first discuss how to compute the infimum of $\letterCost(w)$ over all words accepted by $\aut$, i.e., $\inf_{w \in \lang(\aut)} \letterCost(w)$.
Next, we solve the average letter cost problem by computing $\inf_{w \in \lang(\aut)} \letterCostLambda(w)$ for a modified letter cost function $\letterCostLambda$.

\begin{restatable}{lemmaStatement}{MinLetterCost}
Given a PDA $\aut$ and a letter-cost function $\letterCost$, we can compute $\inf_{w \in \lang(\aut)} \letterCost(w)$, the infimum of $\letterCost(w)$ over all words accepted by $\aut$,  
 in polynomial time in $\aut$. 
\label{l:minLetterCost}
\end{restatable}
\noindent \emph{Remark}. The values of the letter-cost function in Lemma~\ref{l:minLetterCost} can be represented in binary.

\begin{proof}
\noindent\textbf{Overview.}
To compute $\inf_{w \in \lang(\aut)} \letterCost(w)$, we transform $\aut$ to a CFG  $G$ generating the same language. 
Then, we adapt the classic algorithm for checking the emptiness of the language generated by a CFG~\cite{HU79}.
The algorithm from~\cite{HU79} marks iteratively non-terminals that derive some words.
It starts by marking non-terminals that derive a single letter. Then it
takes $|G|$ iterations of a loop, in which it applies all the rules of $G$, and marks non-terminals that derive marked non-terminals.
Here, we associate with each non-terminal $A$, a variable $v_A$ storing the minimal value of $\letterCost(w)$ for $w$ derivable from $A$ in $G$.
We update values of $v_A$ in each iteration, by putting $v_A = \min(v_A, v_B+v_C)$ for every rule $A \to BC$.
The algorithm terminates if (1)~there is an iteration, in which no variable has changed or
(2)~after $|G|+1$ iterations. In the first case, we return $v_S$, the computed value for the start symbol. 
In the second case, we observe that no further iterations are necessary as there exists a derivation 
$B \to_{G}^{*} v_L B v_R$ with $\letterCost(v_L v_R) < 0$, and hence $\inf_{w \in \lang(\aut)} \letterCost(w)= -\infty$.

\Paragraph{Detailed proof} We transform the PDA $\aut$ into an equivalent context-free grammar $G$ in the Chomsky normal form. 
The transformation takes polynomial time (in $\aut$) and $G$ has polynomial size in~$\aut$. 
We can assume that the grammar $G$ is pruned, i.e., every non-terminal $A$ can be derived from $S$ and $A$ derives some word.
Let $N$ be the number of non-terminals in $G$.

We associate with each non-terminal symbol $A$, a variable $v_A$ storing the (current) minimal value of $\letterCost(w)$ for $w$ derivable from $A$ in $G$.
Initially, we put $v_S = 0$ if $S \to \epsilon$ is a production of $G$. Also, for every 
non-terminal $A$ we define $v_A$ as the minimum over $\letterCost(a)$ such that $A \to a$ is a production in $G$.
Next, we iterate $N+1$ times the following procedure, 
for every production rule $A \to BC$, set $v_A$ to $min(v_A, v_B + v_C)$.
Then, if none of the variables $v_A$ has changed in the last iteration, the algorithm returns $v_S$ as the value of $\inf_{w \in \lang(\aut)} \letterCost(w)$.
Otherwise, it returns $-\infty$.

The algorithm takes $N+1$ iterations and each iteration takes $|G|$ steps, and hence it works in polynomial time.
For correctness, observe that $n$-th iteration of the loop examines derivation trees of height $n$. 
Now, assume that in some iteration no variable $v_A$ changes. Then, further iterations of the loop will not change any of the values of $v_A$, and each of 
$v_A$ stores the minimal value of $\letterCost(w)$ for $w$ derivable from $A$ in $G$.
Therefore, if in the last iteration no variable $v_A$ has changed, then $v_S$ equals $\inf_{w \in \lang(\aut)} \letterCost(w)$.

Assume that in the last iteration, for some non-terminal $A$ the value of $v_A$ changes to $l$.
Consider a minimal derivation tree $d$ with the root $A$ such that the letter cost of the derived word is at most $l$.
We know that the height of $d$ is at least $N+1$ and hence $d$ has a path with some non-terminal $B$ occurring at least twice. 
Then, $G$ has a derivation $B \to_{G}^{*} v_L B v_R$, which corresponds to the part of $d$ with both occurrences of  $B$.
Observe that $\letterCost(v_L v_R) < 0$. Indeed, if  $\letterCost(v_L v_R) \geq 0$, then the corresponding part of $d$ can be removed and the letter cost of the derived word 
does not increase. This violates the minimality of $d$.
Finally, as $G$ is pruned, there are word $x_1, x_2, x_3$ such that  $S \to_{G}^{*} x_1 B x_3$ and $B \to_{G}^{*} x_2$.
It follows that for all $i>0$, we have $x_1 u_L^i x_2 u_R^i x_3 \in \lang(\aut)$ and
$\letterCost(x_1 u_L^i x_2 u_R^i x_3 ) < \letterCost(x_1 x_2 x_3) - i$. Thus, 
$\inf_{w \in \lang(\aut)} \letterCost(w) = \infty$. 
\end{proof}

Using Lemma~\ref{l:minLetterCost}, we can solve the average letter cost problem in the finite word case.

\begin{restatable}{lemmaStatement}{FiniteAvgLetterCost}The $\avgLC$ problem can be solved in polynomial time.
\label{l:avgLetterCost}
\end{restatable}

\begin{proof}\textbf{Overview.}
First, if $\inf_{w \in \lang(\aut)} \avgLC(w) < \lambda$, then there is a word $w$ with $\avgLC(w) < \lambda$.
Observe that the average of $a_1, \ldots, a_n$ is less than lambda if and only if the sum of 
$(a_1 -\lambda), \ldots, (a_n -\lambda)$ is less than $0$. Therefore, to check existence of  $w$ with $\avgLC(w) < \lambda$,
we define $\letterCostLambda$ by substructing $\lambda$ from each value of $\letterCost$ and apply 
Lemma~\ref{l:minLetterCost} to check existence of $w$ with   $\letterCostLambda(w) < 0$.
The case of the non-strict inequality is more difficult as the infimum need not be realized. 
However , we consider a CFG $G$ generating the language of $\aut$ and we show that
 $\inf_{u \in \lang(\aut)} \avgLC(u) \leq \lambda$  if and only if 
 (1)~there exists $u \in \lang(\aut)$ with $\avgLC(u) \leq \lambda$, or
 (2)~there is a non-terminal $A$ such that $A \to_{G}^{*} u_L A u_R$ and $\avgLC(u_L u_R) \leq \lambda$.
 Both conditions can be checked in polynomial time using the letter-cost function $\letterCostLambda$ and Lemma~\ref{l:minLetterCost}.

\Paragraph{Detailed proof} 
\Paragraph{Case when $\bowtie=<$}
Observe that $\inf_{u \in \lang(\aut)} \avgLC(u) < \lambda$ if and only if there exists $u \in \lang(\aut)$ such that $\avgLC(u) < \lambda$.
To check the latter, 
 we define a letter-cost function $\letterCostLambda$ by substructing $\lambda$ from each value of $\letterCost$, i.e., 
for $a \in \Sigma$ we put $\letterCostLambda(a) = \letterCost(a) - \lambda$. 
Observe that for every $u$ we have $\avgLC(u) < \lambda$ if and only if 
$\letterCostLambda(u) < 0$. 
Due to Lemma~\ref{l:minLetterCost}, we can decide in polynomial time whether there exists $u \in \lang(\aut)$ with $\letterCostLambda(u) < 0$.
 
\Paragraph{Case when $\bowtie=\leq$}
To check whether $\inf_{u \in \lang(\aut)} \avgLC(u) \leq \lambda$ we construct a CFG $G$ generating $\lang(\aut)$.
We assume that $G$ is pruned, i.e., every non-terminal occurs in some derivation of some word.
Observe that 
 $\inf_{u \in \lang(\aut)} \avgLC(u) \leq \lambda$  if and only if 
 (1)~there exists $u \in \lang(\aut)$ with $\avgLC(u) \leq \lambda$, or
 (2)~there is a non-terminal $A$ such that $A \to_{G}^{*} u_L A u_R$ and $\avgLC(u_L u_R) \leq \lambda$.
For the implication from right to left note that
(1)~implies  $\inf_{u \in \lang(\aut)} \avgLC(u) \leq \lambda$ and (2) implies that there exists a sequence of words 
$u_i = x_1 u_L^i x_2 u_R^i x_3$, and $\lim_{i \to \infty} \avgLC(u_i) = \avgLC(u_L u_R) \leq \lambda$.

We show the implication from left to right by contraposition. Assume that (1) and (2) do not hold.
We say that a word is \emph{prime} if it has a derivation tree in which no non-terminal occurs more than once along the same path; 
the length of a prime word in exponentially bounded in the grammar.
We define $\epsilon$ as the minimum $\avgLC(u) - \lambda$ over all prime words $u$ 
and $\avgLC(u_L u_R) - \lambda$ over all derivations $A \to_{G}^{*} u_L A u_R$ where $u_L, u_R$ are prime.
Since (1) and (2) do not hold and there are finitely many prime words, we have $\epsilon > 0$. 
Assume that  $\inf_{u \in \lang(\aut)} \avgLC(u) \leq \lambda$  and there exists a word $z$ with   $\avgLC(z) \leq \lambda + \frac{\epsilon}{2}$.
We pick $z$ to have the minimal length. If $z$ is not prime, then there is a non-terminal $A$ that occurs at lease twice along some path in some derivation tree of $z$, i.e.,
$z$ can be partitioned into $x_1 u_L x_2 u_R x_3$, where $A \to_{G}^{*} u_L A u_R$. We can pick $A$ and its occurrences in such a way that 
$u_L, u_R$ are prime. Then, $\avgLC(x_1 x_2 x_3) \leq \avgLC(z)$, which contradicts minimality of the length $z$.
Therefore, $z$ is prime. But then, we picked $\epsilon$ so that $\avgLC(z) \geq \lambda + \epsilon$ and we have a contradiction 
with  $\avgLC(z) \leq \lambda + \frac{\epsilon}{2}$.

Finally, conditions (1) and (2) can be checked in polynomial time. 
To check (1), we proceed as in the strict case, i.e., we observe that $\avgLC(u) \leq \lambda$ if and only if $\letterCostLambda(u) \leq 0$.
The latter condition can be check in polynomial time (Lemma~\ref{l:minLetterCost}).
To check (2), for every non-terminal $A$ of $G$, we define a CFG $G_A$ such that $\lang(G_A) = \set{u_L u_R \mid A \to_{G}^* u_L A u_R }$. 
Such a CFG $G_A$ can be constructed in polynomial time. Next, we need to check whether 
there exists $u \in \lang(G_A)$ with $\avgLC(u) \leq \lambda$, which can be done in polynomial time as for condition~(1). 
\end{proof}

\subsection{Average letter cost over infinite words}
\label{s:infiniteAWC}

In this section, we discuss the average letter cost problem in the infinite-word case. For any $\omega$-PDA $\aut$ we can construct in polynomial time PDA recognizing non-empty languages $V_1, U_1, \ldots, V_k, U_k$ such that $\lang(\aut)=\bigcup_{1\leq i \leq k} V_i (U_i)^{\omega}$ and $k \leq |\aut|$~\cite{cohen-gold}. We will call $V_1, U_1, \ldots, V_k, U_k$ \emph{a factorisation} of $\aut$.
We begin with the average letter cost problem with the limit supremum of partial averages. 

\begin{restatable}{lemmaStatement}{InfiniteSupAvgLetterCost}
The $\avgSupLCprob{<}$ problem can be decided in polynomial time.
\label{l:avgInfLetterCost}
 \end{restatable}

\begin{proof} \textbf{Overview.}
We reduce the problem to the finite-word case and use Lemma~\ref{l:avgLetterCost}.
Consider an instance of the $\avgSupLCprob{<}$ problem consisting of
an $\omega$-PDA $\aut$, letter-cost function $\letterCost$,  $\bowtie \in \set{<, \leq}$ and $\lambda$.
Any context-free omega language $\lang$  can be presented as 
$\lang = \bigcup_{1\leq i \leq k} V_i (U_i)^{\omega}$, where $V_1, U_1, \ldots, V_k, U_k$ are non-empty finite-word context-free languages.
We can look for $w$ in each $V_i (U_i)^{\omega}$ separately.
Then, we show that
there exists a word $w \in V_i (U_i)^{\omega}$ such that $\avgSupLC(w) \bowtie \lambda$ 
if and only if
$\inf_{u \in U_i} \avgLC(u_i) \bowtie \lambda$.
The latter condition can be checked in polynomial time (Lemma~\ref{l:avgLetterCost}).

\Paragraph{Detailed proof} \Paragraph{Case when $\bowtie=<$}.
Assume an $\omega$-PDA $\aut$ and let $V_1, U_1, \ldots, V_k, U_k$ be its factorisation. We focus on one $i \in \{1, \ldots, k\}$ as we can check all components $V_i (U_i)^{\omega}$ independently. 
 
We claim that there exists a word $w \in V_i (U_i)^{\omega}$ such that $\avgSupLC(w) < \lambda$ if and only if
$\inf_{u \in U_i} \avgLC(u_i) < \lambda$. The later can be checked in polynomial due to Lemma~\ref{l:avgLetterCost}.

If $\inf_{u \in U_i} \avgLC(u) < \lambda$, then there exists $u \in \lang(U_i)$ such that $\avgLC(u) < \lambda$. 
Then, for any $v \in V_i$, we have $v u^{\omega} \in \lang(\aut)$ and 
$\avgSupLC(v u^{\omega}) = \avgLC(u ) < \lambda$. 
Conversely, if $\inf_{u \in U_i} \avgLC(u) \geq \lambda$, then, for all $u \in U_i$ we have $\avgLC(u) \geq \lambda$.
Every word $w \in V_i (U_i)^{\omega}$ can be represented as $w = v u_1 u_2 \ldots$, where $v \in V_i$ and all for all $j$ we have $u_j \in U_i$ and $\avgLC(u_j) \geq \lambda$.
The limit supremum of partial averages $\avgLC(v), \avgLC(v u_1), \ldots$ is at least $\lambda$, and hence $\avgSupLC(w) \geq \lambda$.

\medskip
\Paragraph{Case when $\bowtie=\leq$}. This case is very similar. 
We claim that there exists a word $w \in V_i (U_i)^{\omega}$ such that $\avgSupLC(w) \leq \lambda$ if and only if
$\inf_{u \in U_i} \avgLC(u_i) \leq \lambda$; the later can be checked in polynomial time due to Lemma~\ref{l:avgLetterCost}.

For the implication from left to right consider a sequence $u_1, u_2, \ldots \in U_i$ such that $\avgLC(u_i) < \lambda + \frac{1}{i}$.
Let $v \in V_i$. For every $j$, we have $v u_j^{\omega} \in V_i U_i^{\omega}$ and  $\avgSupLC(v u_j^{\omega}) < \lambda + \frac{1}{i}$.
Then, we can find a sequence $i_1, i_2, \ldots \in \N$, which ensures that 
$\avgSupLC(v u_1^{i_1} u_2^{i_2} \ldots ) \leq \lambda$, and we have  $v u_1^{i_1} u_2^{i_2} \ldots \in V_i U_i^{\omega}$.
Basically, we need to ensure that the partial averages do not exceed $\lambda$ in words $u_i$.
We define a discrepancy of $u$, denoted by $disc(u)$ as the difference between the minimal and the maximal average over all prefixes of $u$.
Observe that in word $v u_1^{i_1} u_2^{i_2} \ldots$, the wiggle of partial averages is bounded by the maxim over $j$ of
$disc(u_j) \cdot |u_j|$ divided by the position at which $u_j$. Therefore, we need to find a sequence $i_1, i_2, \ldots \in \N$ such that 
$\lim_{j\to\infty} \frac{disc(u_j) \cdot |u_j|)}{|v u_1^{i_1} u_2^{i_2} \ldots u_{j-1}^{i_{j-1}}|} = 0$. If a sequence $i_1, i_2, \ldots$ grows sufficiently fast, the aforementioned limit is $0$.

The converse implication is virtually the same as in the strict case. If $\inf_{u \in U_i} \avgLC(u_i) > \lambda$, then for some
$\epsilon >0$ every word $u \in U_i$ has $\avgLC(u) > \lambda + \epsilon$, so $\avgSupLC(w) \geq \lambda + \epsilon$.
\end{proof}

Now, we consider the average letter cost problem with the limit infimum of partial averages.
We again reduce it to the finite-word case, but now the  reduction is not as straightforward.
The following example explains the main difficulty. 

\begin{example}
\label{ex:avgInfLC}
Consider $\Sigma = \set{0,2}$, and the letter-cost function $\letterCost$, which simply returns the value of a letter, i.e., for $x \in\Sigma$ we have
$\letterCost(x) = x$. Now, let $\aut$ be an $\omega$-PDA accepting the language $U_0^{\omega}$, where $U_0 = \set{0^n 2^n \mid n \in \N}$.
Observe that for every $w \in \lang(\aut)$ we have $\avgSupLC(w) = 1$.
However, for a word $w_0 = 02 0^2 2^2 \ldots 0^{2^{n^2}} 2^{2^{n^2}} \ldots$ we observe that
$\avgLC(02 0^2 2^2 \ldots 0^{2^{n^2}}) <  2^{(n-1)^2 + 1 - n^2} = 2^{-2(n-1)}$, and hence
$\avgInfLC(w_0) = 0$. 
Therefore, $\avgSupLC(w_0) \neq \avgInfLC(w_0)$.
We conclude that  for a language of the form $U^{\omega}$, knowing average letter costs of words in $U$ is insufficient to decide whether there is a word $w$
with  $\avgInfLC(w) \leq \lambda$. Still, in the following we show how to decide $\avgInfLC(w) \leq \lambda$ by examining the structure of $U$.
\end{example}

\begin{restatable}{lemmaStatement}{InfiniteInfAvgLetterCost}
The $\avgInfLCprob{<}$ problem can be decided in polynomial time.
\end{restatable}
\begin{proof}
\Paragraph{Overview}
 Again we represent the language of an $\omega$-PDA $\aut$ as 
$\lang = \bigcup_{1\leq i \leq k} V_i (U_i)^{\omega}$, where $V_1, U_1, \ldots, V_k, U_k$ are non-empty finite-word context-free languages and look for $w$ in each $V_i (U_i)^{\omega}$ separately.
Let $G_i$ be a CFG generating the language $U_i$. We assume that $G_i$ is pruned, i.e., every non-terminal occurs in some derivation of some word.
For $\bowtie \in \set{<,\leq}$, we show that there exists a word $w \in V_i (U_i)^{\omega}$ such that $\avgInfLC(w)  \bowtie \lambda$ 
 if and only if
(1)~$\inf_{u \in U_i} \avgLC(u_i) \bowtie \lambda$ or 
(2)~there exists a non-terminal $A$ in $G_i$, such that $\inf \set{ \avgLC(u_L) \mid A \to_{G_i}^* u_L A u_R}  \bowtie \lambda$. 
Both conditions can be checked in polynomial time using Lemma~\ref{l:avgLetterCost}.
Condition (1) is inherited from the $\avgSupLC$.
For condition (2), observe that $\avgInfLC(w)  \leq \lambda$ if there is a subsequence of partial averages that converges to a value at most $\lambda$.
For a word $v u_1  u_2 \ldots \in V_i (U_i)^{\omega}$, the subsequence from the limit infimum may pick only positions inside words $u_1, u_2, \ldots$ (as in Example~\ref{ex:avgInfLC}), and the subsequence 
of partial averages at boundaries of words may converge to a higher value.
Condition (2) covers this case. In Example~\ref{ex:avgInfLC}, $U_0$ is generated by a grammar $S \to 0 S 2, S \to \epsilon$ and observe that 
this grammar satisfies condition (2) with $\lambda = 0$, i.e., for  $S \to 0 S 2$ we have $\avgLC(0) \leq 0$.

\Paragraph{Detailed proof}
\Paragraph{Case when $\bowtie=<$}.
Assume an $\omega$-PDA $\aut$ and let $V_1, U_1, \ldots, V_k, U_k$ be its factorisation.
To check, whether there exists a word $w \in \lang(\aut)$ such that $\avgInfLC(w) < \lambda$ (resp., $\avgInfLC(w) \leq \lambda$), we can check independently all cases of $i \in \{1, \ldots, k\}$. 
Hence, we fix $i \in \{1, \ldots, k\}$ and focus on one component $V_i (U_i)^{\omega}$.

Let $G_i$ be a CFG generating the language $U_i$. We assume that $G_i$ is pruned, i.e., every non-terminal occurs in some derivation of some word.
We claim that there exists a word $w \in V_i (U_i)^{\omega}$ such that $\avgInfLC(w) < \lambda$ if and only if
(1)~$\inf_{u \in U_i} \avgLC(u_i) < \lambda$ or 
(2)~there exists a non-terminal $A$ in $G_i$, such that $\inf \set{ \avgLC(u_L) \mid A \to_{G_i}^* u_L A u_R}  < \lambda$.

To show the implication from right to left we consider the two cases. If (1) holds, then for some $u \in U_i$ we have $\avgLC(u_L) < \lambda$, then, for any $v \in V_i$, we have $v_i u_i^{\omega} \in V_i (U_i)^{\omega}$ and $\avgInfLC(v_i u_i^{\omega}) = \avgLC(u_i) < \lambda$.
If (2) holds, assume $A$ to be a non-terminal of $G_i$ such that there is a derivation $A \to_{G_i}^* u_L A u_R $ with $\avgLC(u_L) < \lambda$. 
Since $G_i$ is pruned, there are derivations $S \to_{G_i}^* x_1 A x_3$ and $A \to_{G_i}^* x_2$. It follows that for $v \in V_i$ and for every sequence of natural numbers
$i_1, i_2, \dots$ the word $v x_1 u_L^{i_1} x_2 u_R^{i_1} x_3 x_1 u_L^{i_2} x_2 u_R^{i_2} x_3 \ldots$ belongs to $V_i (U_i)^{\omega}$.
If we pick the sequence  $i_1, i_2, \dots$ to grow fast enough, we can make the sequence
$\avgLC(v x_1 u_L^{i_1}), \avgLC(v x_1 u_L^{i_1} x_2 u_R^{i_1} x_3 x_1 u_L^{i_2}), \ldots$ (averages of prefixes up to consecutive occurrences of $x_2$) to converge to $\avgLC(u_L) < \lambda$.

To show the implication from left to right,
 assume towards contradiction that there exists
 $w \in V_i (U_i)^{\omega}$ such that $\avgInfLC(w) < \lambda$, and the following two conditions hold:\\
(C1)~for all $u \in U_i$ we have $\avgLC(u_L) \geq \lambda$\\
(C2)~for all non-terminals $A$ in $G_i$ and all derivations $A \to_{G_i}^* u_L A u_R $ we have $\avgLC(u_L) \geq \lambda$.
Let $\letterCostLambda$ be defined as $\letterCostLambda(a) = \letterCost(a) - \lambda$ for every $a \in \Sigma$. 
(C1) an (C2) imply that:\\
(L1)~for all non-terminals $A$ in $G_i$ and all derivations $A \to_{G_i}^* u_L A u_R $  we have $\letterCostLambda(u_L u_R) \geq 0$.\\
(L2)~for all non-terminals $A$ in $G_i$ and all derivations $A \to_{G_i}^* u_L A u_R $ we have $\letterCostLambda(u_L) \geq 0$.
(*)~there exists $N >0$ such that for every
$u \in U_i$, the minimal value $\letterCostLambda(u')$ among all prefixes $u'$ of $u$ is greater than $-N$.

Implications (C1) $\Rightarrow$ (L1) and (C2) $\Rightarrow$ (L2) are straightforward; we focus on how (L1) and (L2) imply (*).
Consider $B >0$ and a shortest word $u \in U_i$ such that for some prefix $u[1,l]$ of $u$ we have 
 $\letterCostLambda(u[1,l]) < -B$. Consider a derivation tree $d$ for $u$ and mark a path $\sigma$ from the root to the position $l$.
 Condition (L2) implies that if there is
  non-terminal $A$ that occurs twice along $\sigma$, then we can substitute the larger derivation tree rooted at earlier occurrence of $A$ by its subtree rooted at preceding occurrence of $A$, and
 in the resulting word the minimum among prefixes does not increase. 
More precisely,  we can present $u$ as $x_1 u_L x_2 x_3 u_R x_4$ and $u[1,l]$ as $x_1 u_L x_2$, where 
$A \to_{G_i}^* u_L A u_R$ and $A \to_{G_i}^* x_2x_3$. Condition (L2) states that $\letterCostLambda(u_L) \geq 0$, and hence
$x_1 x_2$ is a prefix of $x_1x_2 x_3x_4 \in U_i$, and $\letterCostLambda(x_1 x_2) \leq \letterCostLambda(u[1,l])$.
This contradicts the minimality of length of $u$. Therefore, the length of path $\sigma$ is bounded by $|G_i|$.
 Similarly, using condition (L1) we can show that along paths in $d$ branching of $\sigma$, all non-terminals are different, and hence these paths have length bounded by $|G_i|$.
 Now, we conclude that the derivation tree $d$ of has paths of length bounded by $2 |G_i|$ and hence $u$ is exponentially bounded in $|G_i|$.
 It follows that $B$ is exponentially bounded in $|G_i|$ as well. This shows (*).
 
 Now, the word $w$ can be presented as $v u_1 u_2 \ldots$, where $v \in V_i$ and for all $j$ we have $u_j \in U_i$.
Observe that  $\avgInfLC_{\letterCost}(w) < \lambda$ implies that  $\liminf_{p \to \infty} \letterCostLambda(w[1,p]) = -\infty$. However,
condition (*) implies that at every position $p$ in $w$ we have $\letterCostLambda(w[1,p]) \geq -N +\letterCostLambda(v)$, a contradiction.
Indeed, let $p$ be a position. Then, $w[1,p] = v u_1 u_2 \ldots u_{m-1} u_m'$, for some $m$ and some prefix $u_m'$ of $u_m$. 
Condition (C1) implies that  $\letterCostLambda(u_1 u_2 \ldots, u_{m-1}) > 0$, and hence $ \letterCostLambda(w[1,p]) \geq \letterCostLambda(v u_{m}') \geq \letterCostLambda(v) +  \letterCostLambda(u_{m}')$.
By condition (*), we have $\letterCostLambda(u_{m}') < -N$.

It remains to discuss how to check conditions (1) and (2). 
Condition (1) can be check in polynomial time (Lemma~\ref{l:avgLetterCost}).
To check (2), we construct,  for every non-terminal $A$, a CFG $G_L^A$ such that $u \in \lang(G_L^A)$ iff $G_i$ has a derivation $A \to_{G_i}^* u A u' $ for some word $u'$.
CFG $G_L^A$ can be constructed in polynomial time in $G_i$~\cite{editDistance}. Then, using Lemma~\ref{l:avgLetterCost}, we check in polynomial time whether 
there exists $u_L  \in \lang(G_L^A)$ such that $\avgLC(u_L) <0$.
  
\Paragraph{Case when $\bowtie=\leq$}.
 We proceed as in the strict case. We decompose $\aut$ into $V_1, U_1, \dots, V_k, U_k$
and focus on $V_i (U_i)^{\omega}$. We claim that 
 there exists a word $w \in V_i (U_i)^{\omega}$ such that $\avgInfLC(w) \leq \lambda$ if and only if
(1)~$\inf_{u \in U_i} \avgLC(u_L) \leq \lambda$ or
(2)~there exists a non-terminal $A$ in $G_i$, such that $\inf \set{ \avgLC(u_R) \mid A \to_{G_i}^* u_L A u_R}  \leq \lambda$.

To show the implication from right to left we consider two cases.
If (1) holds, then there is a sequence of words $u_1, u_2, \ldots \in U_i$ such that $\lim_{j \to \infty} \avgLC(u_j) \leq \lambda$.
Then, for any $v \in V_i$, we have 
$v u_1 u_2 \ldots \in V_i U_i^{\omega}$ and $ \avgInfLC(v u_1 u_2 \ldots) \leq \lambda$.
If (2) holds, then there exist derivations $A \to_{G_i}^* u_{L,1} A u_{R,1}, A \to_{G_i}^* u_{L,2} A u_{R,2}, \ldots$ such that
$\lim_{j \to \infty} \avgLC(u_{L,j}) \leq \lambda$.  
Let $x_1, x_2, x_3$ be words such that 
$S \to_{G_i}^* x_1 A x_2$ and $A \to_{G_i}^* x_2$ and let $v \in V_i$.

Then, the word 
$v x_1 u_{L,1}^{i_1} x_2 u_{R,1}^{i_1} x_3 x_1 u_{L,1}^{i_2} x_2 u_{R,2}^{i_2} x_3  \ldots$ belongs to $V_i (U_i)^{\omega}$ for every sequence of natural numbers
$j_1, \dots$.
If we pick the sequence  $i_1, i_2, \dots$ to grow fast enough, we can make the sequence
$\avgLC(v x_1 u_{L,1}^{i_1}), \avgLC(v x_1 u_{L,1}^{i_1} x_2 u_{R,1}^{i_1} x_3 x_1 u_{L,1}^{i_2}), \ldots$ (averages of prefixes up to consecutive occurrences of $x_2$) 
to converge to $\avgLC(u_L) < \lambda$.

To show the implication from left to right, assume that (1) and (2) do not hold. Then, there exists $\epsilon >0$, such that
(C1')~for all $u \in U_i$ we have $\avgLC(u_L) \geq \lambda + \epsilon$, and 
(C2')~for all non-terminals $A$ in $G_i$ and all derivations $A \to_{G_i}^* u_L A u_R $ we have $\avgLC(u_L) \geq \lambda + \epsilon$.
By the proof of the strict case, conditions (C1') and (C2') imply that  for all $w \in V_i (U_i)^{\omega}$ we have $\avgInfLC(w) \geq \lambda + \epsilon > \lambda$

\end{proof}

\subsection{Weighted pushdown systems with fairness}
\label{s:pushdown-systems}

\newcommand{\wt}{\mathbf{wt}}

\newcommand{\limavginf}{\mathsf{LimAvgInf}}
\newcommand{\limavgsup	}{\mathsf{LimAvgSup}}
\newcommand{\seqLetterCost}{\textsf{seq}\letterCost}

We briefly discuss the connection between the average letter cost problem and one-player games on WPSs with 
conjunctions of mean-payoff and \buchi{} objectives.

A \emph{WPS-game} consists of a WPS $\WPS = (\aut, \WPScost)$ and a \emph{game objective}.
In each WPS-game, the only player plays infinitely many rounds selecting consecutive transitions in order to obtain a run satisfying given objectives.
A game objective is a conjunction of a \emph{mean-payoff} objective and a \buchi{} objective, defined as follows.

A \emph{mean-payoff} objective is of the form 
$\limavginf(\run) \bowtie \lambda$ or 
$\limavgsup(\run) \bowtie \lambda$, where $\bowtie \in \set{<, \leq}$ and $\lambda \in \Q$. The interpretation of such an objective is as follows: each play constructs a run $\run$ of $\WPS$ and the \emph{cost sequence} $\wt(\run)$ of $\run$ which is the sequence of costs of the transitions of $\run$.  We interpret $\limavginf(\run)$ as $\liminf_{k\to\infty} \frac{1}{k} \sum_{i=1}^{k} \wt(\run)[i]$ and
 $\limavgsup(\run)$ as $\limsup_{k\to\infty} \frac{1}{k} \sum_{i=1}^{k} \wt(\run)[i]$, and we say that a mean-payoff objective is satisfied if its inequality holds.
A \buchi{} objective is a set of states; it is satisfied by $\run$ if $\run$ visits some state from $Q_F$ infinitely often.

To \emph{solve a WPS-game} is to determine whether the player can construct a run that satisfies all the game objectives. The following theorem extends  \cite[Theorem~1]{ChV12} by adding \buchi{} objectives.

\begin{restatable}{theorem}{Games}
Each WPS-game can be solved in polynomial time.
\end{restatable}

The proof follows from a reduction to the average letter cost problem that encodes the transitions costs in corresponding letter costs. A converse polynomial-time reduction is also possible; in this case, we encode letter costs in transition costs.

\begin{proof}
We claim that WPSs with \buchi{} conditions 
and $\omega$-PDA with letter cost functions 
 are polynomial-time equivalent with respect to the sets of weight sequences.
More precisely, 
for a letter cost function and a word $w$, we define $\seqLetterCost(w)$ as the sequence of costs of consecutive letters, i.e.,
$\seqLetterCost(w) = (\letterCost(w[1]),\letterCost(w[2]), \ldots)$. For every WPS~$\WPS$ and $Q_F$, we can construct in polynomial time an $\omega$-PDA $\aut$ and a letter cost function $\letterCost$ such that 
(*)~$\set{\wt(\run) \mid  \run \text{ satiafies the \buchi{} condition } Q_F  } = \set{\seqLetterCost(w) \mid w \in \lang(\aut)}$. 
The construction follows the idea of the proof of Theorem~\ref{t:reduction}. The converse transformation exists as well and it takes polynomial time; 
given $\aut$ and $\letterCost$ it suffices to define cost of transitions as the cost of the corresponding letter and then erase the letters from transitions
the resulting WPS with accepting states $Q_F$ of $\aut$ satisfy (*).
In consequence, we have the following theorem, which generalizes the results from~\cite{ChV12} by allowing additional \buchi{} objectives.
\end{proof}

 \section{Average Response Time Example}
\label{s:example}

In this section, we use the ASC problem to compute a variant of the average response time property~\cite{nested}.
In this variant, there are two agents: a \emph{client} and a \emph{server}. A client can state a request, which is later granted or rejected by the server. 
Requests are dealt with on the first-come, first-served basis, but not immediately --- the server may need some time to issue a grant. 
We assume that both client and server are modeled as systems with finitely many states and can check whether the number of pending requests at a given moment is zero. We also assume a fairness condition stating that there are infinitely many requests and grants.

A trace of such system is a word over the alphabet $\set{r, g, \#}$, where $r$ denotes a new request, $g$ denotes a grant and $\#$ denotes a null instruction. 
 We are interested in bounding the minimal possible average response time of such a model
. In other words, we are interested checking, for a given $\lambda$ and a model, whether 
\begin{equation}\label{e:art}
 \liminf_{n \to \infty} \frac{1}{n}\sum_{i=1}^n a_i < \lambda 
\end{equation} 
  for some computation of the model in which the $i$th request was realised after $a_i$ steps of computations (our technique works for $\limsup$ as well, but we focus on $\liminf$). 

\Paragraph{Feasibility study} To apply our technique, we need to overcome two main difficulties. First, our technique only works for stacks, but requests are handled in a queue manner. In general, non-emptiness of automata with queue is undecidable. Second, the denominator in \eqref{e:art} refers only to the number of requests, not the number of positions in words (they may differ because of the letter $\#$).

\Paragraph{Dealing with queues}
We abstract the counter to a stack over a unary alphabet $\set{P}$ whose size equals the value of the counter in the straightforward way.

We claim that there is a run satisfying \eqref{e:art} if and only if there is a run satisfying 
\begin{equation}\label{e:art-eq}
 \liminf_{n \to \infty} \frac{1}{n}\sum_{i=1}^{G_n} \cost(\run[i]) < \lambda 
\end{equation}
where $G_n$ denotes the position of the $n$th grant in the run (we assume that each grant correspond to a request). 

We first discuss the main idea. 
Consider a position in a run $n$ where there are no pending requests; then, the total waiting time of all requests up to this position is equal to the sum of the number of waiting processes in each position up to $n$. At a position with unfulfilled requests, this is no longer guaranteed, as the pending processes may have some waiting time in the future. However, it can be shown that a run for \eqref{e:art-eq} can be chosen in a way that guarantees that the difference between the two numbers is bounded by some constant, and therefore can be neglected in the $\liminf$.

Let us briefly recall that for a trace $w$ over the alphabet $\set{r,g,\#}$ we define
$a_i$ as the number of steps between $i$-th request and its corresponding grant, 
$R_n$ (resp., $G_n$) is the position of the $i$-th request (resp., $i$-th grant), and
$\cost(\run[i])$ is the number of pending requests at the position $i$.
We show that if 
 $\liminf_{n \to \infty} \frac{1}{n}\sum_{i=1}^n a_i$ is finite, then
\[
 \liminf_{n \to \infty}  \frac{1}{n}\sum_{i=1}^n a_i =  \liminf_{n \to \infty} \frac{1}{n}\sum_{i=1}^{G_n} \cost(\run[i])\;.
 \]

Assume that $\liminf_{n \to \infty} \frac{1}{n}\sum_{i=1}^n a_i = B < \infty$. Observe that
\[
 \liminf_{n \to \infty} \frac{1}{n}\sum_{i=1}^n a_i \leq  \liminf_{n \to \infty} \frac{1}{n}\sum_{i=1}^{G_n} \cost(\run[i])\;.
 \]
To see that, note that in $\sum_{i=1}^n a_i$ we count each positions until $R_n$ as many times as
there are pending requests, and for positions between $R_n$ and $G_n$, we count the number of pending requests issued up to position $R_n$ (we ignore requests issued past $R_n$).

\sloppy
Now, we prove that 
\[ \liminf_{n \to \infty} \frac{1}{n}\sum_{i=1}^n a_i \geq  \liminf_{n \to \infty} \frac{1}{n}\sum_{i=1}^{G_n} \cost(\run[i])\;.
\]
By the above discussion we know that $\frac{1}{n}\sum_{i=1}^n a_i \geq  \frac{1}{n}\sum_{i=1}^{R_n} \cost(\run[i])$, and hence we show that
$\liminf_{n \to \infty}  \frac{1}{n}\sum_{i=R_n+1}^{G_n} \cost(\run[i]) = 0$.
Consider $\epsilon > 0$.
Since  $\liminf_{n \to \infty} \frac{1}{n}\sum_{i=1}^n a_i= B$, there are infinitely many $n$'s such that 
(1)~$\frac{1}{n}\sum_{i=1}^n a_i \leq B + \epsilon$, and
(2)~$a_n < B+1$. 
Condition (2) means that $n$-th request has been answered in less than $B+1$ steps, i.e., $G_n - R_n < B+1$.
Since we consider the queue for requests, there can be at most $B$ pending requests at position $R_n$, i.e., 
 $\cost(\run[R_n]) < B+1$. Finally, $\cost(\run[i])$ changes by at most one at every step and hence
 for all $i \in \set{R_n, \ldots, G_n}$ we have $ \cost(\run[i]) < 2B$. 
 Therefore, there are infinitely many $n$'s such that  $\sum_{i=R_n+1}^{G_n} \cost(\run[i]) \leq 2B^2$, and hence
 $\liminf_{n \to \infty}  \frac{1}{n}\sum_{i=R_n+1}^{G_n} \cost(\run[i])= 0$.

\Paragraph{Selected positions} 
We argue that $\eqref{e:art-eq}$ is equivalent to 
\begin{equation}\label{e:artall}
\liminf_{n \to \infty} \frac{1}{n}\sum_{i=1}^{n} \cost^*(\run[i]) < \lambda 
\end{equation}
 where $\cost^*(\run[i])$ equals $\cost(\run[i])$ if the position $i$ corresponds to a grant and $\cost(\run[i])+\lambda$ otherwise.

Observe that
\(\frac{1}{n}\sum_{i=1}^{G_n} \cost(\run[i]) < \lambda \) iff \(\sum_{i=1}^{g_n} \cost(\run[i]) - n\lambda + g_n \lambda< g_n \lambda \) iff 
\(\frac{1}{g_n}\sum_{i=1}^{G_n} \cost^*(\run[i]) < \lambda \). The last equivalence follows from the fact that there are $n$ grants, and so $\sum_{i=1}^{G_n} \cost^*(\run[i]) = \sum_{i=1}^{G_n} \cost(\run[i]) + (G_n-n) \lambda$.

If \eqref{e:art-eq}, then there is an infinite sequence of positions where \(\frac{1}{n}\sum_{i=1}^{G_n} \cost(\run[i]) < \lambda \), and by the above reasoning each position in this sequence satisfies 
\(\frac{1}{G_n}\sum_{i=1}^{G_n} \cost^*(\run[i]) < \lambda \). 
This means that \eqref{e:art-eq} implies \eqref{e:artall}. The converse if also true. 
To see this, observe that if at a position $n>0$ that does not correspond to a grant we have \(\frac{1}{n}\sum_{i=1}^{n} \cost^*(\run[i]) < \lambda \), then also \(\frac{1}{n-1}\sum_{i=1}^{n-1} \cost^*(\run[i]) < \lambda \) as $\cost^*(\run[n]) \geq \lambda$. 
If we have an infinite sequence of positions with \(\frac{1}{n}\sum_{i=1}^{n} \cost^*(\run[i]) < \lambda \) and infinitely many grants, we can select an infinite sequence of positions $n$ corresponding to grants such that \(\frac{1}{n}\sum_{i=1}^{n} \cost^*(\run[i]) < \lambda \).

\medskip

\Paragraph{Putting it all together} From the above consideration, we know that \eqref{e:art} if and only if \eqref{e:artall}. Therefore, to verify \eqref{e:art}, we modify the automaton as follows: we add an additional stack symbol $\bullet$ of weight $\lambda+1$ that can only appear at the top of the stack. We modify the transition function to stipulate that whenever the automaton is in a position that does not correspond to a grant, then the topmost symbol is $\bullet$. By our results, checking whether there is a run with the average stack cost less than $\lambda$ (and therefore whether the average waiting time is less than $\lambda$) can be done in polynomial time.
 
\bibliography{all}
\end{document}